\documentclass[preprint,12pt,authoryear]{elsarticle}

\usepackage{amssymb}
\usepackage{amsthm}
\usepackage{graphicx}
\usepackage{marvosym}
\usepackage{amsmath}
\usepackage{amsfonts}
\usepackage{pdfsync}
\newtheorem{definition}{Definition}
\newtheorem{theorem}{Theorem} 
\newtheorem{lemma}{Lemma}     

\newtheorem{remark}{Remark}

\journal{Computers \& Industrial Engineering}

\begin{document}

\begin{frontmatter}



\title{Facility Location Game with Envy Ratio}

\author[a]{Yuan Ding}\ead{dingyuan@ouc.edu.cn}
\author[a]{Wenjing Liu\corref{cor1}}\ead{liuwj@ouc.edu.cn}
\author[a]{Xin Chen}\ead{cxin0307@163.com}
\author[a]{Qizhi Fang}\ead{qfang@ouc.edu.cn}
\author[a]{Qingqin Nong}\ead{qqnong@ouc.edu.cn}

\cortext[cor1]{Corresponding author.}


\address[a]{School of Mathematical Sciences, Ocean University of China, Qingdao 266100, China}

\begin{abstract}
We study the one-facility location game on a real line with a new objective called envy ratio. The envy ratio, which is adopted from fair division and represents the egalitarianism, is defined as the maximum over the ratios between any two agents' utilities. We are interested in strategyproof or group strategyproof mechanisms that can minimize the envy ratio objective.

We consider the model in two settings that can capture natural scenarios: the facility location and all the agents' locations are restricted on a fixed interval; every agent's location can be any point on the real line but the facility location is restricted on a relative interval. In both settings, we obtain the optimal solution and the best deterministic strategyproof mechanism which is also group strategyproof. In the first setting, we provide a lower bound for randomized strategyproof mechanisms. In the second setting, we give a lower bound and two upper bounds for randomized strategyproof mechanisms.
\end{abstract}




\end{frontmatter}


\section{Introduction}

\textit{Approximate mechanism design without money} for facility location games has been extensively studied in the last ten years since the pathbreaking contribution of \citet{Procaccia2009}. In the basic setting, the social planner (e.g., the local government) plans to design a mechanism which locates a facility (e.g., a library or a bus station) based on the reported locations from self-interested agents. Every agent has a connection cost which is usually her distance from the facility. On the one hand, the social planner wants to optimize a certain objective, such as minimizing the sum of individual costs (i.e., the social cost). On the other hand, every agent who holds her location as private information, is self-interested and can misreport to decrease her own connection cost. Naturally, the social planner needs to design mechanisms which can optimize the objective while ensuring every agent's truthful report (i.e., strategyproof or group strategyproof).

Besides physical aspects, the facility location setting can be extended to many other applications. For example, every agent's location can represent her qualitative political view and the facility location can be a political decision. In the foregoing scenarios, money or payment cannot be used as a medium of compensation due to ethical or legal considerations \citep{Schummer2007}. Moreover, considering that the optimal solution may be not strategyproof, the goal of the social planner becomes to design strategyproof mechanisms without money that have small approximation ratios. Our work belongs to the research agenda of approximate mechanism design without money.

Most of the existing work studies the objectives of minimizing the social cost (or equivalently, maximizing the social welfare) or minimizing the maximum cost \citep{Procaccia2009,Lu2009,Lu2010,Cheng2013,FT2014,Zhang2014,SV2014,SV2015,Zou2015,Yuan2016,Feigenbaum2017,Chen2018,Fong2018,Duan2019,Li2019,Mei2019}. From the perspective of economics, the social cost objective can represent the utilitarianism. By contrast, the maximum cost objective is often considered to be more fair and can represent the egalitarianism, which is more important than the utilitarianism in many environments.

In this paper, we consider the facility game with a new objective called \textit{envy ratio}, which can also represent the egalitarianism and is a more direct fairness criterion than the maximum cost. Our objective of the envy ratio is motivated by fair division, which has been a central topic in the economy theory. Several concepts of fairness have been suggested and one of them is envy-free allocation, which means that every player prefers her own share to the share of any other player \citep{F1967,V1974}. \citet{Lipton2004} studied the problem of allocating indivisible goods with minimum possible envy from an algorithmic perspective and one objective of the minimization problems is the envy ratio. Analogically, in the facility location setting, we define the \textit{envy ratio} as the maximum over the ratios between any two agents' utilities.

We study the one-facility location game on a real line with the objective of minimizing the envy ratio.
Our goal is to design strategyproof or group strategyproof mechanisms with small approximation ratios.
\subsection{Our Results}

This paper studies strategyproof or group strategyproof approximation mechanisms for the one-facility location game on a real line with the objective of minimizing the envy ratio.
We consider the problem in two settings which can capture many real life scenarios. In the first setting, the locations of all the agents and the single facility are restricted on a fixed interval. This setting models the scenarios where the facility location and the location preferences of agents are constrained by some physical boundary. For example, when planning to build a library for some community, it is reasonable to assume that the location of the library and all the community members' location preferences are in the community. In the second setting, every agent's location can be any point on the real line but the facility location is restricted by the reports of the agents. This setting models the scenarios where the facility should not be too far away from every agent due to some consideration.

Our key innovation and results are summarized as follows.

In Section \ref{pre}, we formulate the one-facility location game with envy ratio. To the best of our knowledge, it is the first attempt to consider the envy ratio as the objective function of the facility location games. Our work contributes in the fields of approximate mechanism design without money and fair allocation.

In Section \ref{fixed}, we study the facility location setting where every agent's location and the facility location are restricted on a fixed interval. We obtain the optimal solution which is not strategyproof. For the deterministic case, we give a lower bound of 2 for any strategyproof mechanism and design a group strategyproof mechanism with approximation ratio of 2, which implies that the best deterministic strategyproof mechanism has been given. For the randomized case, we show a lower bound of 1.0314 for any strategyproof mechanism.

In Section \ref{relative}, we study the setting where every agent's location can be any point on the real line but the facility location is restricted on an interval related to the location profile $\mathbf{x}$, denoted as $[lm(\mathbf{x})-\beta L(\mathbf{x}),rm(\mathbf{x})+\beta L(\mathbf{x})]$. Here, $\beta>0$ is a parameter, $lm(\mathbf{x})$ is the leftmost point of $\mathbf{x}$, $rm(\mathbf{x})$ is the rightmost point of $\mathbf{x}$ and $L(\mathbf{x})=rm(\mathbf{x})-lm(\mathbf{x})$ is the length of $\mathbf{x}$. In this setting, we also provide the optimal solution which is not strategyproof. For the deterministic case, we show a lower bound of $1+1/\beta$ for any strategyproof mechanism and provide a group strategyproof mechanism with approximation ratio of $1+1/\beta$, which implies that the best deterministic strategyproof mechanism has been obtained. For the randomized case, when $\beta>1/2$, we give a lower bound of $1+\frac{\displaystyle 2\beta-1}{\displaystyle 8\beta^2(1+\beta)}$ for any strategyproof mechanism. As for the upper bound, we propose a randomized group strategyproof mechanism with approximation ratio of $\displaystyle 1+1/(2\beta)$. Furthermore, if $\beta\ge 1$, we show a randomized strategyproof mechanism which has a smaller approximation ratio of $\displaystyle 1+2/(1+\beta)^2$.

\subsection{Related Work}\label{related}

Mechanism design for the facility location game has a rich history of research. \citet{Moulin1980} characterized all strategyproof, efficient and anonymous mechanisms for the facility location game with the single peaked preference on the real line. \citet{Schummer2002} provided a complete characterization of strategyproof mechanisms on other networks. However, these works focus on the characterization of strategyproof mechanisms and do not consider the optimizations or approximations over a certain objective.

Approximate mechanism design without money for facility location games was formally initiated by \citet{Procaccia2009}. They studied the facility location game on the real line with the social cost objective and the maximum cost objective in three settings. In the one-facility setting, they gave the optimal and group strategyprooof mechanism for the social cost objective and the best strategyproof mechanism for the maximum cost objective. For the social cost objective, they gave an upper bound of $n-2$ and a lower bound of 3/2 for deterministic strategyproof mechanisms in the two-facility setting. They also studied strategyproof mechanisms in the two-facility setting with the maximum cost objective and in the setting of locating one facility but every agent having multiple locations.

Since then, approximate strategyproof mechanism design for facility location games has been well studied. For the two-facility location game with the social cost objective, \citet{Lu2009} provided an upper bound of $n/2$ and a lower bound of 1.045 for randomized strategyproof mechanisms. \citet{Lu2010} improved the lower bound for deterministic strategyproof mechanisms to $(n-1)/2$ and proposed a randomized strategyproof 4-approximation mechanism. \citet{FT2014} gave an elegant characterization of deterministic strategyproof mechanisms, which provided a lower bound of $n-2$.

Many variants of the problem have also been proposed to accommodate more scenarios. For the one-facility location game, \citet{Cheng2013} proposed an obnoxious facility game on networks where every agent wants to stay far away from the facility. \citet{Zhang2014} extended the model to games with weighted agents on a line segment. Further, \citet{Feigenbaum2015} and \citet{Zou2015} studied the dual preference game where some agents want to stay close to the facility while the others want to stay away from the facility. In addition, \citet{Mei2019} introduced a happiness factor to measure the agent's degree of satisfaction for the facility location and \citet{Li2019} studied the facility game with externalities where every agent's utility is affected by others. For the two-facility location game, research has been extended from the original homogeneous facility setting to the heterogeneous facility setting where facilities serve different purposes \citep{SV2014,SV2015,Zou2015,Yuan2016,Fong2018,Duan2019}.

The concerned objectives in the foregoing literature fall into two major categories: the social cost and the maximum cost, with the only exception of the work of \citet{Cai2016}, which has motivated our work in some sense. They studied the facility location game with the objective of minimizing the maximum envy. The maximum envy, which is also adopted from the fair division literature \citep{Lipton2004}, is defined as the maximum difference over all the agents' distances from the facility normalized by the length of the location profile. Considering that the optimal value is always 0 for instances with only two different locations, they use an additive approximation to measure the performance of approximate mechanisms. They generalized the classic LRM mechanism proposed by \citet{Procaccia2009} and gave a class of strategyproof mechanisms, the best of which has performance arbitrarily close to the optimal solution.

\section{Preliminaries}
\label{pre}

Let $N=\{1,2,\cdots,n\}$ be a set of agents, where each agent $i$ has a location $x_i\in \mathbb{R}$, which is $i$'s private information. If the facility is located at $y\in \mathbb{R}$, the cost of agent $i\in N$ is defined as the distance between the facility and agent $i$, that is, $cost(y,x_i)=|x_i-y|$.

We refer to the collection $\mathbf{x}=(x_1,x_2,\cdots,x_n)\in \mathbb{R}^n$ as a location profile or an instance. For $i\in N$, let $\mathbf{x}_{-i}=(x_1,\cdots,x_{i-1},x_{i+1},\cdots,x_n)$, then $\mathbf{x}=(x_i,\mathbf{x}_{-i})$. For a nonempty subset $S$ of $N$, let $\mathbf{x}_S={(x_i)}_{i\in S}$, $\mathbf{x}_{-S}={(x_i)}_{i\notin S}$, then $\mathbf{x}=(\mathbf{x}_S,\mathbf{x}_{-S})$.
For a location profile $\mathbf{x}$, if there are $K$ different locations $x_1,\cdots,x_K$ and $N$ can be partitioned into $K$ coalitions $N_1,\cdots,N_K$ such that all agents in $N_i$ occupy a same location $x_i$, $\mathbf{x}$ is referred as a $K$-location instance. We denote such an instance as $(x_1:N_1,\cdots,x_K:N_K)$ \citep{FT2014}.

\textbf{Mechanisms.} A deterministic mechanism is a function $f: \mathbb{R}^n\rightarrow \mathbb{R}$, which maps a location profile to a facility location. A randomized mechanism is a function which maps a location profile to a probability distribution over the facility locations. Formally, a randomized mechanism is a function $f: \mathbb{R}^n\rightarrow \Delta(\mathbb{R})$, where $\Delta(\mathbb{R})$ is the set of probability distributions over $\mathbb{R}$.

Given a deterministic (or randomized) mechanism and an instance $\mathbf{x}\in \mathbb{R}^n$, the cost of agent $i\in N$ is $cost(f(\mathbf{x}),x_i)=|f(\mathbf{x})-x_i|$ (or $\mathbb{E}_{Y\sim f(\mathbf{x})}|Y-x_i|$).

\textbf{Strategyproofness and (partial) Group Strategyproofness.} A mechanism $f$ is \textit{strategyproof} if no agent can benefit from misreporting her location, regardless of the other agents' strategies. Formally, for every location profile $\mathbf{x}\in \mathbb{R}^n$, every agent $i\in N$, and every $x_i'\in \mathbb{R}$, $cost(f(x_i',\mathbf{x}_{-i}),x_i)\ge cost(f(\mathbf{x}),x_i)$. A mechanism $f$ is \textit{group strategyproof} if for any coalition of agents misreporting their locations, at least one of them can not benefit. Formally, for every location profile $\mathbf{x}$, every coalition of agents $S\subseteq N$, and every $\mathbf{x}_S'\in \mathbb{R}^{|S|}$, there exists some agent $i\in S$ such that $cost(f(\mathbf{x}_S',\mathbf{x}_S),x_i)\ge cost(f(\mathbf{x}),x_i)$. A mechanism is \textit{partial group strategyproof} if for any coalition of agents that occupy the same location, none of them can benefit from misreporting their locations simultaneously. Formally, for every location profile $\mathbf{x}$, every coalition of agents $S$ occupying a same location $x$, and every $\mathbf{x}_S'\in \mathbb{R}^{|S|}$, $cost(f(\mathbf{x}_S',\mathbf{x}_{-S}),x)\ge cost(f(\mathbf{x}),x)$.

\begin{remark}
By the definition, any \textit{group strategyproof} mechanism is also \textit{partial group strategyproof}, and any \textit{partial group strategyproof} mechanism is also \textit{strategyproof}. Furthermore, it has been showed that any \textit{strategyproof} mechanim is also \textit{partial group strategyproof} in the facility location games \citep{Lu2010}. Thus, we will not distinguish between \textit{strategyproof} and \textit{partial group strategyproof} in the following analysis.
\end{remark}

\textbf{Envy Ratio.} We are interested in minimizing the \textit{envy ratio} of all agents. For a location profile $\mathbf{x}$, the \textit{envy ratio} of the facility location $y$ is defined as

\begin{equation}
ER(y,\mathbf{x})=\max_{1\le i\neq j\le n}\frac{u(y,x_i)}{u(y,x_j)},
\end{equation}
where $u(y,x_i)$ is the utility of agent $i$ with respect to $y$ and will be defined specifically in the next two sections. Notice that $ER(y,\mathbf{x})\ge 1$.

For a location profile $\mathbf{x}$, let $OPT(\mathbf{x})$ be the optimal solution to the optimization problem $min_y ER(y,\mathbf{x})$ and $ER(OPT,\mathbf{x})$ be the optimal envy ratio.

The \textit{envy ratio} of a deterministic (or randomized) mechanism $f$ is defined as $ER(f,\mathbf{x})=ER(f(\mathbf{x}),\mathbf{x})$ (or $\mathbb{E}_{Y\sim f(\mathbf{x})}ER(Y,\mathbf{x})$).

Next, we give the definition of \textit{approximation ratio} which was advocated by \citet{Procaccia2009} to measure the performance of a mechanism, and some notations which will be used in the following analysis.

\textbf{Approximation Ratio.} A mechanism $f$ is said to have an approximation ratio of $\rho$ ($\rho\ge 1$), if it satisfies
\begin{equation}
\rho=\sup_{\mathbf{x}} \frac{ER(f,\mathbf{x})}{ER(OPT,\mathbf{x})}.
\end{equation}

In this paper, our goal is to design strategyproof or group strategyproof mechanisms with minimum possible approximation ratios.

\textbf{Notations.} For a location profile $\mathbf{x}$, denote $lm(\mathbf{x})=\min_{i\in N}x_i$ which is the leftmost point of $\mathbf{x}$, $rm(\mathbf{x})=\max_{i\in N}x_i$ which is the rightmost point of $\mathbf{x}$, $mid(\mathbf{x})=1/2(lm(\mathbf{x})+rm(\mathbf{x}))$ which is the midpoint of $\mathbf{x}$ and $L(\mathbf{x})=rm(\mathbf{x})-lm(\mathbf{x})$ which is the length of $\mathbf{x}$.

\section{Envy Ratio on a Fixed Interval}
\label{fixed}

In this section, we restrict every agent $i$'s location $x_i$ and the facility location $y$ on a fixed interval $[0,L]$. The utility of agent $i\in N$ is defined as $u(y,x_i)=L-cost(y,x_i)$.

For a location profile $\mathbf{x}\in{[0,L]}^n$, the envy ratio of the facility location $y\in [0,L]$ can be written as
\begin{equation}
ER(y,\mathbf{x})=\max_{1\le i\neq j\le n}\frac{L-cost(y,x_i)}{L-cost(y,x_j)}.
\end{equation}

We will give an accurate characterization of the optimal solution for the envy ratio.

\begin{lemma}\label{mid}
Given a location profile $\mathbf{x}$, the facility location $mid(\mathbf{x})$ minimizes the envy ratio.
\end{lemma}

\begin{proof}
Let $\mathbf{x}$ be a location profile, we now consider the monotonicity of $ER(y,\mathbf{x})$ as a function of $y$.

Without loss of generality, assume that $L(\mathbf{x})>0$ and $\mathbf{x}$ is a $K$-location instance denoted by $(x_1:N_1,\cdots,x_K:N_K)$ with $x_1<x_2<\cdots<x_K$, we first consider the monotonicity of $ER(y,\mathbf{x})$ on the interval $[0,mid(\mathbf{x})]$.

For $y\in[0,mid(\mathbf{x})]$, we only need to analyze the following three cases.

\textbf{Case 1.} If $y\in[0,x_1)$, $ER(y,\mathbf{x})=\frac{\displaystyle L-|y-x_1|}{\displaystyle L-|y-x_K|}$, which will decrease as $y$ increases.

\textbf{Case 2.} If $y\in\left(x_i,\frac{\displaystyle x_i+x_{i+1}}{\displaystyle 2}\right)$ for some $i\in\{1,2,\cdots,K\}$, $ER(y,\mathbf{x})=\frac{\displaystyle L-|y-x_i|}{\displaystyle L-|y-x_K|}$. It is obvious that $ER(y,\mathbf{x})$ decreases as $y$ increases.

\textbf{Case 3.} If $y\in\left(\frac{\displaystyle x_i+x_{i+1}}{\displaystyle 2},x_{i+1}\right)$ for some $i\in\{1,2,\cdots,K\}$, $ER(y,\mathbf{x})=\frac{\displaystyle L-|y-x_{i+1}|}{\displaystyle L-|y-x_K|}$, which will decrease as $y$ increases.

Combining with the continuity of $ER(y,\mathbf{x})$ at the points $x_i$ and $\frac{\displaystyle x_i+x_{i+1}}{\displaystyle 2}$, it follows that $ER(y,\mathbf{x})$ is monotonically decreasing on $[0,mid(\mathbf{x})]$.

Through a symmetric analysis, we can obtain that $ER(y,\mathbf{x})$ is monotonically increasing on $[mid(\mathbf{x}),L]$.

It is clear that $ER(y,\mathbf{x})$ has its minimum at the point $mid(\mathbf{x})$. Furthermore, $mid(\mathbf{x})$ is the unique point that minimizes $ER(y,\mathbf{x})$ for any instance $\mathbf{x}$ with $L(\mathbf{x})>0$.
\end{proof}

Unfortunately, the facility location of $mid(\mathbf{x})$ is not strategyproof. For instance $\mathbf{x}=(L/4:N_1,L/2:N_2)$, where $N_1=\{1\}, N_2=\{2,\cdots,n\}$, agent 1 can move the facility location to her own location by reporting $0$ instead of her true location $L/4$. The facility location problem in Section \ref{relative} also faces the same challenge. Next, we will aim at seeking strategyproof or group strategyproof mechanisms with minimum possible approximation ratio. The monotonicity property of the envy ratio $ER(\cdot,\mathbf{x})$ given in the proof of Theorem \ref{mid} will be very useful for analyzing the approximation ratio of a mechanism.

\subsection{Deterministic Mechanisms}

In this subsection, we consider deterministic strategyproof mechanisms and the approximation ratio for the envy ratio.

\begin{theorem}
The approximation ratio of any deterministic strategyproof mechanism is at least 2.
\end{theorem}

\begin{proof}
Let $f$ be any deterministic strategyproof mechanism. According to the equivalence between strategyproof and partial group strategyproof in the facility location game, $f$ is also partial group strategyproof.

Consider a 2-location instance $\mathbf{x}=(0:N_1,L:N_2)$, where $N_1=\{1\},N_2=\{2,\cdots,n\}$. Note that $OPT(\mathbf{x})=L/2$ and $ER(OPT,\mathbf{x})=1$.

If $f(\mathbf{x})=0$ (or $L$), $ER(f,\mathbf{x})=+\infty$.

Otherwise, $f(\mathbf{x})\in(0,L)$. Without loss of generality, assume that $f(\mathbf{x})=L/2+\epsilon$, where $\epsilon\in[0,L/2)$. Let $\mathbf{x}'=(0:N_1,L/2+\epsilon:N_2)$. By $f$'s (partial group) strategyproofness, $f(\mathbf{x}')=L/2+\epsilon$; otherwise, agents in coalition $N_2$ with the location $L/2+\epsilon$ can benefit by misreporting to location $L$ simultaneously. Thus, $ER(f,\mathbf{x}')=\displaystyle\frac{L-0}{L-(L/2+\epsilon)}=\frac{L}{L/2-\epsilon}\ge 2$.

Anyhow, $f$ has an approximation ratio of at least 2.
\end{proof}

\begin{theorem}
$f(\mathbf{x})=L/2$ is a group strategyproof 2-approximation mechanism for the envy ratio.
\end{theorem}

\begin{proof}
$f$ is group strategyproof since it does not depend on any information from the agents. We only need to prove the mechanism has an approximation ratio of 2.

For any location profile $\mathbf{x}\in{[0,L]}^n$, $\displaystyle ER(f,\mathbf{x})\le\frac{L-0}{L-L/2}=2$ and $ER(OPT,\mathbf{x})\ge 1$. Thus, $\displaystyle\sup_{\mathbf{x}} \frac{ER(f,\mathbf{x})}{ER(OPT,\mathbf{x})}\le 2$.

Consider a 2-location instance $\mathbf{x}'=(0:N_1,L/2:N_2)$, where $N_1=\{1\}$, $N_2=\{2,\cdots,n\}$.
Note that $OPT(\mathbf{x}')=L/4$, $ER(OPT,\mathbf{x}')=1$ and $\displaystyle ER(f,\mathbf{x}')=\frac{L-0}{L-L/2}=2$, which implies that $\displaystyle\frac{ER(f,\mathbf{x}')}{ER(OPT,\mathbf{x}')}=2$.

Therefore, $f$ has an approximation ratio of 2.
\end{proof}

The above analysis demonstrates that any deterministic strategyproof mechanism has an approximation ratio of at least 2, and simply locating at the midpoint of the fixed interval is exactly the best deterministic strategyproof mechanism, which may be somewhat surprising. Then, a following question is whether the randomization can break the deterministic lower bound of 2 or not, which will be discussed in the next subsection.

\subsection{Randomized Mechanisms}

\begin{theorem}
Any randomized strategyproof mechanism has an approximation ratio of at least 1.0314.
\end{theorem}

\begin{proof}
Let $f$ be any randomized strategyproof mechanism.
Consider a location profile $\mathbf{x}=(L/2:N_1,3L/4:N_2)$, where $N_1=\{1\},N_2=\{2,\cdots,n\}$.
According to the triangle inequality, $cost(f(\mathbf{x}),L/2)+cost(f(\mathbf{x}),3L/4)=\mathbb{E}_{Y\sim f(\mathbf{x})}[|Y-L/2|+|Y-3L/4|]\ge |L/2-3L/4|=L/4$. Then, either $cost(f(\mathbf{x}),3L/4)\ge L/8$ or $cost(f(\mathbf{x}),L/2)\ge L/8$. We will analyze the approximation ratio of $f$ through these two cases.

\textbf{Case 1.} $cost(f(\mathbf{x}),L/2)\ge L/8$. Consider $\mathbf{x}'=(L/4:N_1,3L/4:N_2)$. Then we have
\begin{equation}\label{eqn7}
cost(f(\mathbf{x}'),L/2)\ge cost(f(\mathbf{x}),L/2)\ge L/8.
\end{equation}

Otherwise, agent 1 at location $L/2$ can benefit by reporting $L/4$, which contradicts $f$'s strategyproofness.

Let $Y'$ be a random variable according to the probability distribution $f(\mathbf{x}')$.
Denote $\delta=(4-\sqrt{7})L/12$, $p_1=Pr\{Y'\in[L/2-\delta,L/2+\delta]\}$, $p_2=Pr\{Y'\in[L/4,L/2-\delta)\cup(L/2+\delta,3L/4]\}$, and $p_3=Pr\{Y'\in[0,L/4)\cup(3L/4,L]\}$, we have
\begin{eqnarray}\label{eqn8}
cost(f(\mathbf{x}'),L/2)&=&\mathbb{E}[|Y'-L/2|]\nonumber\\
&=& p_1\mathbb{E}\left[|Y'-L/2|~|Y''\in[L/2-\delta,L/2+\delta]\right]\nonumber\\
&&{}+p_2\mathbb{E}\left[|Y'-L/2|~|Y''\in[L/4,L/2-\delta)\cup(L/2+\delta,3L/4]\right]\nonumber\\
&&{}+p_3\mathbb{E}\left[|Y'-L/2|~|Y''\in[0,L/4)\cup(3L/4,L]\right]\nonumber\\
&\le&\delta p_1+L/4\cdot p_2+L/2\cdot p_3\nonumber\\
&=&\delta+(L/4-\delta)p_2+(L/2-\delta)p_3
\end{eqnarray}

Combining (\ref{eqn7}) with (\ref{eqn8}), we have
\begin{equation}\label{eqn9}
p_3\ge\frac{L/8-\delta-(L/4-\delta)p_2}{L/2-\delta}
\end{equation}

On the other hand, we have
\begin{eqnarray*}
ER(f,\mathbf{x}')&=&p_1\mathbb{E}\left[ER(Y',\mathbf{x}')|Y'\in[L/2-\delta,L/2+\delta]\right]\\
&&{}+p_2\mathbb{E}\left[ER(Y',\mathbf{x}')|Y'\in[L/4,L/2-\delta)\cup(L/2+\delta,3L/4]\right]\\
&&{}+p_3\mathbb{E}\left[ER(Y',\mathbf{x}')|Y'\in[0,L/4)\cup(3L/4,L]\right]\\
&\ge&p_1+\frac{L-(L/4-\delta)}{L-(L/4+\delta)}p_2+\frac{L-0}{L-L/2}p_3\\
&=&1+\frac{2\delta}{3L/4-\delta}p_2+p_3\\
&\ge&1+\frac{2\delta}{3L/4-\delta}p_2+\frac{L/8-\delta-(L/4-\delta)p_2}{L/2-\delta}\\
&=&1+\frac{L/8-\delta}{L/2-\delta}\thickapprox1.0314
\end{eqnarray*}
Here, the first inequality holds because $ER(\cdot,\mathbf{x}')$ monotonically decreases on $[0,L/2]$ and monotonically increases on $[L/2,L]$. The second inequality holds due to (\ref{eqn9}). The last equality holds because $\delta=(4-\sqrt{7})L/12$.

\textbf{Case 2.} $cost(f(\mathbf{x}),3L/4)\ge L/8$. The analysis of this case is similar to that of Case 1. For the completeness, we give the proof here. In this case, consider $\mathbf{x}''=(L/2:N_1,L:N_2)$. By $f$'s strategyproofness, we have that
\begin{equation}\label{eqn4}
cost(f(\mathbf{x}''),3L/4)\ge cost(f(\mathbf{x}),3L/4)\ge L/8.
\end{equation}

Let $Y''$ be a random variable according to the probability distribution $f(\mathbf{x}'')$.

Without confusion, still denote $p_1=Pr\{Y''\in[2L/3,5L/6]\}$, $p_2=Pr\{Y''\in[L/2,2L/3)\cup(5L/6,L]\}$, and $p_3=Pr\{Y''\in[0,L/2)\}$, we have
\begin{eqnarray}\label{eqn5}
cost(f(\mathbf{x}''),3L/4)&=&\mathbb{E}[|Y''-3L/4|]\nonumber\\
&=& p_1\mathbb{E}\left[|Y''-3L/4|~|Y''\in[2L/3,5L/6]\right]\nonumber\\
&&{}+p_2\mathbb{E}\left[|Y''-3L/4|~|Y''\in[L/2,2L/3)\cup(5L/6,L]\right]\nonumber\\
&&{}+p_3\mathbb{E}\left[|Y''-3L/4|~|Y''\in[0,L/2)\right]\nonumber\\
&\le&L/12\cdot p_1+L/4\cdot p_2+3L/4\cdot p_3\nonumber\\
&=&L/12+L/6\cdot p_2+2L/3\cdot p_3
\end{eqnarray}

Combining (\ref{eqn4}) with (\ref{eqn5}), we have
\begin{equation}\label{eqn6}
p_3\ge\frac{1/24-1/6\cdot p_2}{2/3}
\end{equation}

On the other hand, we have
\begin{eqnarray*}
ER(f,\mathbf{x}'')&=&p_1\mathbb{E}\left[ER(Y'',\mathbf{x}'')|Y''\in[2L/3,5L/6]\right]\\
&&{}+p_2\mathbb{E}\left[ER(Y'',\mathbf{x}'')|Y''\in[L/2,2L/3)\cup(5L/6,L]\right]\\
&&{}+p_3\mathbb{E}\left[ER(Y'',\mathbf{x}'')|Y''\in[0,L/2)\right]\\
&\ge&p_1+\frac{L-L/6}{L-L/3}p_2+\frac{L-0}{L-L/2}p_3\\
&=&1+0.25p_2+p_3\\
&\ge&1+0.25p_2+\frac{1/24-1/6\cdot p_2}{2/3}\\
&=&\frac{17}{16}=1.0625
\end{eqnarray*}
Here, the first inequality holds because $ER(\cdot,\mathbf{x}'')$ monotonically decreases on $[0,3L/4]$ and monotonically increases on $[3L/4,L]$. The second inequality holds due to (\ref{eqn6}).

Note that $ER(OPT,\mathbf{x}')=1$ and $ER(OPT,\mathbf{x}'')=1$.
Therefore, $f$ has an approximation ratio of at least 1.0314.
\end{proof}

\subsection{Discussion}\label{discussion1}

For the deterministic case, our results are completely tight. For the randomized case, a lower bound is obtained. However, we failed in the attempt to find any randomized strategyproof mechanism with an approximation ratio less than 2.
Besides, it is not hard to verify that any mechanism which locates at $lm(\mathbf{x})$ or $rm(\mathbf{x})$ with positive probability can not have a bounded approximation ratio.

An interesting open problem is how to narrow the gap between the upper bound of 2 and the randomized lower bound of 1.0314 for the envy ratio on the fixed interval.

\section{Envy Ratio on a Relative Interval}\label{relative}

In this section, every agent's location can be any point on the real line but the facility location is restricted on an interval relative to the location profile. Formally, for a location profile $\mathbf{x}\in \mathbb{R}^n$, the facility location $y$ is restricted on $[lm(\mathbf{x})-\beta L(x),rm(\mathbf{x})+\beta L(x)]$, where $\beta>0$ is a parameter. The utility of agent $i\in N$ is defined as $u(y,x_i)=(1+\beta)L(\mathbf{x})-cost(y,x_i)$\footnote{The utility function is defined due to the consideration that the maximum cost of any agent is no more than $(1+\beta)L(\mathbf{x})$.}.

For a location profile $\mathbf{x}\in \mathbb{R}^n$, the envy ratio of the facility location $y\in[lm(\mathbf{x})-\beta L(x),rm(\mathbf{x})+\beta L(x)]$ can be written as
\begin{equation}
ER(y,\mathbf{x})=\max_{1\le i\neq j\le n}\frac{(1+\beta)L(\mathbf{x})-cost(y,x_i)}{(1+\beta)L(\mathbf{x})-cost(y,x_j)}.
\end{equation}

Now let us turn to the optimal solution for the envy ratio on the relative interval, which is similar to the case of the fixed interval.

\begin{lemma}
Given a location profile $\mathbf{x}$, the facility location $mid(\mathbf{x})$ minimizes the envy ratio.
\end{lemma}

\begin{proof}
For a given location profile $\mathbf{x}$, $L(\mathbf{x})$ is fixed. Through an analysis similar to that of Lemma \ref{mid}, we can show that $ER(y,\mathbf{x})$ is monotonically decreasing on $[lm(\mathbf{x})-\beta L(\mathbf{x}),mid(\mathbf{x})]$, is monotonically increasing on $[mid(\mathbf{x}),rm(\mathbf{x})+\beta L(\mathbf{x})]$ and has the minimum at $mid(\mathbf{x})$.
\end{proof}

The facility location of $mid(\mathbf{x})$ is not strategyproof. Next, we will discuss the approximate strategyproof mechanisms.

\subsection{Deterministic Mechanisms}

In this subsection, we will give a complete characterization of the deterministic strategyproof mechanisms.

\begin{theorem}
The approximation ratio of any deterministic strategyproof mechanism is at least $1+1/\beta$.
\end{theorem}
\begin{proof}
Let $f$ be any deterministic strategyproof mechanism.

Consider a 2-location instance $\mathbf{x}=(0:N_1,1:N_2)$, where $N_1=\{1\},N_2=\{2,\cdots,n\}$. Note that $ER(OPT,\mathbf{x})=1$ and $f(\mathbf{x})\in[-\beta,1+\beta]$.

If $-\beta\le f(\mathbf{x})\le 0$ or $1\le f(\mathbf{x})\le 1+\beta$, then $ER(f,\mathbf{x})\ge ER(0,\mathbf{x})=\frac{\displaystyle 1+\beta-0}{\displaystyle 1+\beta-1}=1+1/\beta$. It follows that $\frac{\displaystyle ER(f,\mathbf{x})}{\displaystyle ER(OPT,\mathbf{x})}\ge 1+1/\beta$.

Otherwise, if $f(\mathbf{x})\in(0,1)$, assume without loss of generality that $f(\mathbf{x})=1/2+\epsilon$, $\epsilon\in [0,1/2)$, and consider a new profile $\mathbf{x}'=(0:N_1,1/2+\epsilon:N_2)$.
Note that $ER(OPT,\mathbf{x}')=1$. By $f$'s strategyproofness, $f(\mathbf{x}')=1/2+\epsilon$. Thus, $ER(f,\mathbf{x}')= \frac{\displaystyle (1+\beta)(1/2+\epsilon)-0}{\displaystyle (1+\beta)(1/2+\epsilon)-(1/2+\epsilon)}=1+1/\beta$. It follows that $\frac{\displaystyle ER(f,\mathbf{x}')}{\displaystyle ER(OPT,\mathbf{x}')}= 1+1/\beta$.

Therefore, $f$ has an approximation ratio of at least $1+1/\beta$.
\end{proof}

\begin{theorem}
$f(\mathbf{x})=lm(\mathbf{x})$ (or $rm(\mathbf{x})$) is a group strategyproof $(1+1/\beta)$-approximation mechanism for the envy ratio.
\end{theorem}

\begin{proof}
We only need to prove the approximation ratio of the mechanism $f(\mathbf{x})=lm(\mathbf{x})$.

For any location profile $\mathbf{x}\in \mathbb{R}^n$ ($L(\mathbf{x})> 0$), $ER(f,\mathbf{x})= \frac{\displaystyle (1+\beta)L(\mathbf{x})-0}{\displaystyle (1+\beta)L(\mathbf{x})-L(\mathbf{x})}=1+1/\beta$, $ER(OPT,\mathbf{x})\ge 1$. Thus, $\sup_{\mathbf{x}}\frac{\displaystyle ER(f,\mathbf{x})}{\displaystyle ER(OPT,\mathbf{x})}\le 1+1/\beta$.

Consider a location profile $\mathbf{x}'=(0:N_1,1:N_2)$, where $N_1=\{1\},N_2=\{2,\cdots,n\}$.
Note that $ER(OPT,\mathbf{x}')=1$ and $ER(f,\mathbf{x}')=\frac{\displaystyle(1+\beta)-0}{\displaystyle(1+\beta)-1}=1+1/\beta$, which implies that $\frac{\displaystyle ER(f,\mathbf{x}')}{\displaystyle ER(OPT,\mathbf{x}')}=1+1/\beta$.

Thus, $f(\mathbf{x})=lm(\mathbf{x})$ has an approximation ratio of $1+1/\beta$.
\end{proof}

Observe that the best deterministic strategyproof mechanism has been obtained and we will turn our attention to randomized strategyproof mechanisms.

\subsection{Randomized Mechanisms}

In this subsection, we will give a lower bound and two upper bounds for randomized strategyproof mechanisms.

\begin{theorem}
If $\beta>1/2$, any randomized strategyproof mechanism has an approximation ratio of at least $1+\frac{\displaystyle 2\beta-1}{\displaystyle 8\beta^2(1+\beta)}$.
\end{theorem}

\begin{proof}
Let $f$ be any randomized strategyproof mechanism.

Consider a location profile $\mathbf{x}=(0:N_1,1:N_2)$, where $N_1=\{1\},N_2=\{2,\cdots,n\}$.
Due to the triangle inequality, $cost(f(\mathbf{x}),0)+cost(f(\mathbf{x}),1)=\mathbb{E}_{Y\sim f(\mathbf{x})}[|Y-0|+|Y-1|]\ge 1$. Without loss of generality, we assume that $cost(f(\mathbf{x}),0)\ge 1/2$ and consider another location profile $\mathbf{x}'=(-1:N_1,1:N_2)$.

We claim that
\begin{equation}\label{eqn1}
cost(f(\mathbf{x}'),0)\ge cost(f(\mathbf{x}),0)\ge 1/2.
\end{equation}

Otherwise, agent 1 at location 0 can benefit by reporting $-1$ instead of 0, which contradicts $f$'s strategyproofness. Let $Y'$ be a random variable according to the probability distribution $f(\mathbf{x}')$. Note that $L(\mathbf{x}')=2$ and $Y'\in [-1-2\beta, 1+2\beta]$.

Denote
$\delta=1/(2\beta+1)$ and let $p_1=Pr\{Y'\in[-\delta,\delta]\},p_2=Pr\{Y'\in[-1,-\delta)\cup(\delta,1]\},p_3=Pr\{Y'\in[-1-2\beta)\cup(1,1+2\beta]\}$.
Then, we have
\begin{eqnarray}\label{eqn2}
cost(f(\mathbf{x}'),0)&=&\mathbb{E}[|Y'-0|]\nonumber\\
&=& p_1\mathbb{E}\left(|Y'|~|Y'\in[\delta,\delta]\right)+p_2\mathbb{E}\left(|Y'|~|Y'\in[-1,-\delta)\cup(\delta,1]\right)\nonumber\\
&&{}+p_3\mathbb{E}\left(|Y'|~|Y'\in[-1-2\beta)\cup(1,1+2\beta]\right)\nonumber\\
&\le&\delta p_1+p_2+(1+2\beta)p_3\nonumber\\
&=&\delta+(1-\delta)p_2+(1-\delta+2\beta)p_3
\end{eqnarray}

Combining (\ref{eqn1}) with (\ref{eqn2}), we have
\begin{equation}\label{eqn3}
p_3\ge\frac{1/2-\delta-(1-\delta)p_2}{1-\delta+2\beta}
\end{equation}

Next, we consider the envy ratio of $f$.
\begin{eqnarray*}
ER(f,\mathbf{x}')&=&p_1\mathbb{E}\left(ER(Y',\mathbf{x}')|Y'\in[-\delta,\delta]\right)+p_2\mathbb{E}(ER(Y',\mathbf{x}')|Y'\in[-1,-\delta)\\
&&{}\cup(\delta,1])+p_3\mathbb{E}\left(ER(Y',\mathbf{x}')|Y'\in[-1-2\beta)\cup(1,1+2\beta]\right)\\
&\ge&p_1+\frac{1+2\beta+\delta}{1+2\beta-\delta}p_2+(1+\frac{1}{\beta})p_3\\
&=&1+\frac{2\delta}{1+2\beta-\delta}p_2+\frac{1}{\beta}p_3\\
&\ge&1+\frac{2\delta}{1+2\beta-\delta}p_2+\frac{1}{\beta}\cdot\frac{1/2-\delta-(1-\delta)p_2}{1-\delta+2\beta}\\
&=&1+\frac{2\beta-1}{8\beta^2(1+\beta)}
\end{eqnarray*}
Here, the first inequality holds because $ER(\cdot,\mathbf{x}')$ monotonically decreases on $[-1-2\beta,0]$ and monotonically increases on $[0,1+2\beta]$. The second inequality holds due to (\ref{eqn3}).

Note that $ER(OPT,\mathbf{x}')=1$, which implies that $f$ has an approximation ratio of at least $1+\frac{\displaystyle 2\beta-1}{\displaystyle 8\beta^2(1+\beta)}$.
\end{proof}

We now turn to seeking randomized strategyproof mechanisms. For this purpose, consider two classes of randomized mechanisms, both of which are very meaningful to the mechanism design for the envy ratio. The first one is a direct generalization of the well-known LRM mechanism \citep{Alon2010,Procaccia2009} and is parameterized by a constant $\gamma\in[0,1/2]$. For convenience, we denote this \textit{generalized LRM mechanism with parameter $\gamma$} as \textit{Mechanism $\gamma$-GLRM} and give the formal definition later. The second one is \textit{Mechanism $\alpha$-LRM}, which is parameterized by a constant $\alpha\in(0,1/4]$, was introduced by \citet{Cai2016}.

Next, we will analyze the approximation ratio and the strategyproofness of these two classes of mechanisms.

\begin{definition}
\textit{Mechanism $\gamma$-GLRM} is parameterized by $\gamma\in[0,1/2]$; for any location profile $\mathbf{x}$, \textit{Mechanism $\gamma$-GLRM} places the facility at $mid(\mathbf{x})$ with probability $1-2\gamma$, at $lm(\mathbf{x})$ with probability $\gamma$ and at $rm(\mathbf{x})$ with probability $\gamma$.
\end{definition}

\begin{lemma}\label{frratio}
\textit{Mechanism $\gamma$-GLRM} has an approximation ratio of $1+2\gamma/\beta$ for the envy ratio.
\end{lemma}

\begin{proof}
Let $f_{\gamma}$ be a mechanism in \textit{Mechanism $\gamma$-GLRM}.
For any location profile $\mathbf{x}\in \mathbb{R}^n$ ($L(\mathbf{x})>0$), $ER(OPT,\mathbf{x})=ER(mid(\mathbf{x}),\mathbf{x})$, and
$ER(f_{\gamma},\mathbf{x})=(1-2\gamma)ER(mid(\mathbf{x}),\mathbf{x})+2\gamma(1+1/\beta)$. It follows that
\begin{equation}
\frac{ER(f_{\gamma},\mathbf{x})}{ER(OPT,\mathbf{x})}=1-2\gamma+\frac{2\gamma(1+1/\beta)}
{ER(mid(\mathbf{x}),\mathbf{x})}\le 1+\frac{2\gamma}{\beta}.
\end{equation}

"=" in the above inequality holds for any 2-location instance. Thus, the approximation ratio of $f_{\gamma}$ is $1+2\gamma/{\beta}$. \qed
\end{proof}

\begin{lemma}\label{frgsp}
\textit{Mechanism $\gamma$-GLRM} is group strategyproof if and only if $\gamma\in[1/4,1/2]$.
\end{lemma}

\begin{proof}
\textbf{If part.} Let $S\subseteq N$ be a coalition. We need to show that the agents in S cannot all gain by deviating. Note that for a given location profile $\mathbf{x}\in \mathbb{R}^n$, $f_{\gamma}(\mathbf{x})$ only depends on the location $lm(\mathbf{x})$ and $rm(\mathbf{x})$. Let $\mathbf{x}'=(\mathbf{x}_S',\mathbf{x}_{-S})$, $\Delta_1=lm(\mathbf{x})-lm(\mathbf{x}')$ and $\Delta_2=rm(\mathbf{x}')-lm(\mathbf{x})$. Now we consider the following cases.

\textbf{Case 1.} $\Delta_1\ge 0$, $\Delta_2\ge 0$. For any $i\in S$,
\begin{eqnarray*}
cost(f_{\gamma}(\mathbf{x}'),x_i&=&\gamma(x_i-lm(\mathbf{x}'))+\gamma(rm(\mathbf{x}')-x_i)\\
&&{}+(1-2\gamma)\left|\frac{lm(\mathbf{x}')+rm(\mathbf{x}')}{2}-x_i\right|\\
&\ge&cost(f(\mathbf{x}),x_i)+\gamma(\Delta_1+\Delta_2)-\frac{1-2\gamma}{2}|\Delta_1-\Delta_2|\\
&\ge&cost(f(\mathbf{x}),x_i)+\gamma(\Delta_1+\Delta_2)-\frac{1-2\gamma}{2}(\Delta_1+\Delta_2)\\
&=&cost(f(\mathbf{x}),x_i)+(2\gamma-1/2)(\Delta_1+\Delta_2)\\
&\ge&cost(f(\mathbf{x}),x_i)
\end{eqnarray*}

\textbf{Case 2.} $\Delta_1< 0$, $\Delta_2\ge 0$. In this case, the leftmost agent must be a member of $S$. It is obvious that this agent cannot benefit from deviating, since the leftmost location, possibly the rightmost location and the center are all moving far away from her.

\textbf{Case 3.} $\Delta_1\ge 0$, $\Delta_2< 0$. In this case, the rightmost agent must be in $S$ cannot benefit from the deviation, which is symmetric to Case 2.

\textbf{Case 4.} $\Delta_1< 0$, $\Delta_2< 0$. In this case, both the leftmost agent and the rightmost agent must be members of $S$.
\begin{eqnarray*}
&&cost(f(\mathbf{x}'),lm(\mathbf{x}))+cost(f(\mathbf{x}'),rm(\mathbf{x}))\\
&=&rm(\mathbf{x})-lm(\mathbf{x})\\
&=&cost(f(\mathbf{x}),lm(\mathbf{x}))+cost(f(\mathbf{x}),rm(\mathbf{x}))
\end{eqnarray*}
Thus, either $cost(f(\mathbf{x}'),lm(\mathbf{x}))\ge cost(f(\mathbf{x}),lm(\mathbf{x}))$, or $cost(f(\mathbf{x}'),rm(\mathbf{x}))\ge cost(f(\mathbf{x}),rm(\mathbf{x}))$. This implies that either the leftmost agent or the rightmost agent cannot benefit from the deviation.

\textbf{Only if part.} We show this part by contradiction. Assume $\gamma< 1/4$ and we will show that there exists an instance such that some agent can benefit from misreporting.

Consider $\mathbf{x}=(0:N_1,1:N_2)$, where $N_1=\{1\}, N_2=\{2,\cdots,n\}$. The cost of agent 1 is $cost(f(\mathbf{x}),0)=1/2$. Let agent 1 misreports her location to $-1$, and denote $\mathbf{x}'=(-1:N_1,1,N_2)$. Then the cost of agent 1 becomes $cost(f(\mathbf{x}'),0)=2\gamma<1/2$, which implies that she can benefit from misreporting. This contradicts $f_{\gamma}$'s strategyproofness.
\end{proof}

Combining Lemma \ref{frratio} with Lemma \ref{frgsp}, we can immediately obtain the following theorem.

\begin{theorem}
\textit{Mechanism 1/4-GLRM} is a group strategyproof mechanism with approximation ratio of $\displaystyle 1+1/(2\beta)$ for the envy ratio.
\end{theorem}

It is obvious that \textit{Mechanism 1/4-GLRM} is the best strategyproof mechanism in the class of \textit{Mechanism $\gamma$-GLRM}. A more intuitive interpretation is provided as follows. Given a location profile $\mathbf{x}$, the outcome of \textit{Mechanism $\gamma$-GLRM} is a randomization between the optimal solution (i.e., $mid(\mathbf{x})$) and the best deterministic strategyproof mechanisms (i.e., $lm(\mathbf{x})$ or $rm(\mathbf{x})$). $mid(\mathbf{x})$ is optimal (i.e., has an approximation ratio of 1) but is not strategyproof, while $lm(\mathbf{x})$ or $rm(\mathbf{x})$ has an approximation ratio of $1+1/\beta$ and is group strategyproof. While increasing the probability $\gamma$ of locating at $lm(\mathbf{x})$ or $rm(\mathbf{x})$, the approximation ratio is gradually sacrificed to achieve strategyproofness. When $\gamma$ increases to 1/4, the optimal tradeoff between approximation ratio and strategyproofness in the class of \textit{Mechanism $\gamma$-GLRM} is obtained.

\begin{definition}\cite{Cai2016}
\textit{Mechanism $\alpha$-LRM} is parameterized by $\alpha\in(0,1/4]$; for any location profile $\mathbf{x}$, denote $L^{\alpha}(\mathbf{x})=\frac{\displaystyle 1-4\alpha}{\displaystyle 4\alpha}L(\mathbf{x})$, \textit{Mechanism $\alpha$-LRM} places the facility at $mid(\mathbf{x})$ with probability $1-2\alpha$, at $lm(\mathbf{x})-L^{\alpha}(\mathbf{x})$ with probability $\alpha$ and at $rm(\mathbf{x})+L^{\alpha}(\mathbf{x})$ with probability $\alpha$.
\end{definition}

\begin{remark}
Considering that for a location profile $\mathbf{x}$, the facility location is restricted on $[lm(\mathbf{x})-\beta L(x),rm(\mathbf{x})+\beta L(x)]$, we restrict the parameter $\alpha$ in $\left(\frac{\displaystyle 1}{\displaystyle 4(1+\beta)},1/4\right]$.
\end{remark}

\begin{lemma}
\textit{Mechanism $\alpha$-LRM} has an approximation ratio of $1+\frac{\displaystyle 2\alpha}{\displaystyle 1+\beta-1/(4\alpha)}$ for the envy ratio.
\end{lemma}

\begin{proof}
Let $f$ be a mechanism in the class of \textbf{$\alpha$-LRM}. For any location profile $\mathbf{x}\in \mathbb{R}^n$ ($L(\mathbf{x})> 0$), $ER(OPT,\mathbf{x})=ER(mid(\mathbf{x}),\mathbf{x})$, $ER(f,\mathbf{x})=(1-2\alpha)ER(mid(\mathbf{x}),
\mathbf{x})+2\alpha\displaystyle\cdot\frac{2+\beta-1/(4\alpha)}{1+\beta-1/(4\alpha)}$. Then
\begin{eqnarray*}
\frac{ER(f,\mathbf{x})}{ER(OPT,\mathbf{x})}&=&1-2\alpha+\frac{\displaystyle 2\alpha\cdot\frac{2+\beta-1/(4\alpha)}{1+\beta-1/(4\alpha)}}{ER(mid(\mathbf{x}),\mathbf{x})}\\
&&\le 1+\frac{2\alpha}{1+\beta-1/(4\alpha)}.
\end{eqnarray*}
"=" in the above inequality holds for any 2-location instance. Thus, the approximation ratio is $\displaystyle 1+\frac{2\alpha}{1+\beta-1/(4\alpha)}$.
\end{proof}

\begin{lemma}{\cite{Cai2016}}
\textit{Mechanism $\alpha$-LRM} is strategyproof.
\end{lemma}

\begin{theorem}
If $\beta\ge 1$, \textit{Mechanism $\displaystyle \frac{1}{2(1+\beta)}$-LRM} is a strategyproof mechanism with approximation ratio of $\displaystyle 1+\frac{2}{(1+\beta)^2}$ for the envy ratio.
\end{theorem}

\begin{remark}
In the class of \textit{Mechanism $\alpha$-LRM}, if $\beta\le 1$, \textit{Mechanism 1/4-LRM} has the optimal approximation ratio of $\displaystyle1+1/(2\beta)$; if $\beta\ge 1$, \textit{Mechanism $\displaystyle\frac{1}{2(1+\beta)}$-LRM} has the optimal approximation ratio of $\displaystyle 1+\frac{2}{(1+\beta)^2}$.
\end{remark}

\subsection{Discussion}\label{discussion2}

 Our results for the deterministic case are completely tight. For the randomized case, there still exists a gap between the upper bound and the lower bound. Indeed, for $\beta=1$, the randomized upper bound given by \textit{Mechanism 1/4-GLRM} (or \textit{Mechanism 1/4-LRM}) is 5/4 and the randomized lower bound is 17/16. How to narrow the gap is an intriguing open problem.

 Moreover, the parameter $\beta$ needs not to be prespecified, which implies that the relative interval setting can be adjusted to various scenarios.

\section{Conclusion and Future Work}\label{future}

The \textit{envy ratio} is a natural fairness criterion adopted from the fair division literature \citep{Lipton2004}. We formulated a one-facility location game with the objective of minimizing the envy ratio, which extends the existing work on facility location games and fair division. We analyzed the problem in two settings where the facility location is restricted on a fixed interval or an interval related to the reported locations of the agents, both of which can apply in many real life scenarios.

For these two settings, we obtained the optimal solutions which are not strategyproof and the best deterministic strategyproof mechanisms. For the randomized strategyproof mechanisms, there exists a gap between the upper bound and the lower bound for the two settings. As we summarized in Section \ref{discussion1} and Section \ref{discussion2}, how to narrow the gap would be an interesting direction.

Our model can be naturally extended to the multiple facility location problem. Besides, it would be interesting to study domains where the space of locations is multi-dimensional Euclidean space, more general metric spaces, or other networks.

\bigskip\noindent\textbf{Declaration of Competing Interest}

\bigskip None.

\bigskip\noindent\textbf{Acknowledgements}

\bigskip This research was supported in part by the National Natural Science Foundation of China (11971447, 11871442), the Natural Science Foundation of Shandong Province of China (ZR2017MD011, ZR2019MA052) and the Fundamental Research Funds for the Central Universities (201964006).






\end{document}